\documentclass{article}
\usepackage{graphicx,amssymb,mathtools,cite,array,multirow,epstopdf,amsthm,comment,bm}
\usepackage[dvipsnames]{xcolor}
\usepackage[ruled,vlined]{algorithm2e}

\newtheorem{lemma}{Lemma}[section]
\theoremstyle{remark}
\newtheorem{remark}{Remark}
\usepackage{geometry}
 \date{}
\begin{document}
\title{Robust Short-term Operation of AC Power Network with Injection Uncertainties}
\author{Anamika Tiwari, Abheejeet Mohapatra, and Soumya Ranjan Sahoo
\thanks{The authors are with the Department of Electrical Engineering, Indian Institute of Technology Kanpur, Kanpur, Uttar Pradesh, India 208016. (email: anamtiw@iitk.ac.in; abheem@iitk.ac.in; srsahoo@iitk.ac.in)}\vspace{-1.7em}}
\maketitle

\begin{abstract}
With uncertain injections from Renewable Energy Sources (RESs) and loads, deterministic AC Optimal Power Flow (OPF) often fails to provide optimal setpoints of conventional generators. A computationally time-efficient, economical, and robust solution is essential for ACOPF with short-term injection uncertainties. Usually, applying Robust Optimization (RO) for conventional non-linear ACOPF results in computationally intractable Robust Counterpart (RC), which is undesirable as ACOPF is an operational problem. Hence, this paper proposes a single-stage non-integer non-recursive RC of ACOPF, using a dual transformation, for short-term injection uncertainties. The proposed RC is convex, tractable, and provides base-point active power generations and terminal voltage magnitudes (setpoints) of conventional generators that satisfy all constraints for all realizations of defined injection uncertainties. The non-linear impact of uncertainties on other variables is inherently modeled without using any affine policy. The proposed approach also includes the budget of uncertainty constraints for low conservatism of the obtained setpoints. Monte-Carlo Simulation (MCS) based participation factored AC power flows validate the robustness of the obtained setpoints on NESTA and case9241pegase systems for different injection uncertainties. Comparison with previous approaches indicates the efficacy of the proposed approach in terms of low operational cost and computation time. \\ \\
\emph{Keywords:} Convexity, Optimal Power Flow, Participation Factor, Robust, Short-Term Operation, Uncertainty.

\end{abstract}

\section{Introduction}
\subsection{Motivation}
Optimal setpoints of conventional generators are obtained every 5-15 min from OPF (or economic dispatch) for secure and reliable AC power network operation \cite{time7436515}. Given the forecasted loads for the upcoming time slot, ACOPF finds setpoints with minimum generation cost to meet the load and satisfy network constraints. The integration of RESs, such as solar and wind generation, has brought about a paradigm shift in the operation of AC power network. Unlike conventional generators, the output of RES is uncertain due to dependence on weather conditions. Also, demand response enablers have increased the unpredictability of loads. Often, the deterministic ACOPF solution cannot ensure economic and secure operation of power network in the presence of uncertain loads and RESs injections \cite{mp1,bienstock2014chance}. \emph{This necessitates designing a robust, time-efficient, optimal dispatch policy to address short-term injection uncertainties so that the operational generation cost in AC network is minimum without any constraint violations}.

\subsection{Literature review}\label{sec1b}
In literature, various methods like Stochastic Programming \cite{sp1,sp2,stoch1}, Chance-Constrained Programming (CCP) \cite{bienstock2014chance,8017474}, Distributionally RO (DRO) \cite{dro1,dro2}, and RO have been proposed to handle uncertainties in OPF. RO only requires the information of possible range of uncertain parameters \cite{robustproof}, and assumes hard constraints\footnote[1]{Soft constraints (whose violations are tolerated if violations` magnitude and duration of violation are not high) in power networks can be considered via CCP \cite{8017474} and DRO \cite{dro1}. However, modeling such constraints is beyond the scope of this paper and shall be explored in future.}, which obviates any violation during the AC network operation with uncertainties \cite{gorissen2015practical}. However, due to the non-convex power flow equations, direct application of RO is strenuous. Hence, the initial robust OPF formulations \cite{6200395,bertismasuc,jabrrobust,ding2016adjustable} use the linear DCOPF model instead of the actual ACOPF model. In DCOPF, voltage magnitudes are all unity, and thus, the obtained solution is not always feasible for the actual AC network \cite{yasasvi}. Further, the presence of uncertain terms in usual non-linear equality constraints of OPF poses a difficulty in directly obtaining a tractable robust ACOPF model \cite{ben2002robust}. In previous works, this challenge is addressed by either converting the infinite-dimensional non-deterministic problem into a finite-dimensional deterministic problem by obtaining its RC \cite{bai,louca,Lee2021}, or by decomposition into multi-stage optimization problem in which the uncertainty is realized in one stage and the other variables are adjusted accordingly in other stages \cite{attarha2018adaptive,lorca,heidarabadi2020robust,yang2021robust}. The associated multi-stage problem is solved iteratively using Cutting Plane (CP) methods.

The RC formulations are computationally efficient but obtaining RC of a problem with uncertain equality constraints is not a straightforward exercise. To exclude uncertain terms in equality constraints, \cite{bai} considers average wind generation instead of actual uncertain generation. In \cite{louca}, the equality constraints are transformed into inequality constraints by assuming adjustable generators to be present at each bus, which is not always true and practically infeasible. The recursive inner-approximated robust ACOPF in \cite{Lee2021} gives conservative solutions with possibility of being infeasible at times due to convex restriction. The initialization of convex restriction is by actual deterministic ACOPF solution, and in each recursion (multi-stage formulation), the states are obtained from AC power flow equations, which increase the computation time. In \cite{attarha2018adaptive}, a tri-stage robust operational problem is posed, where uncertainty is realized in the middle stage and other variables are adjusted via primal and dual cuts. It also approximates actual non-linear power flow equations via linearization for the worst-case uncertainty realization. The formulation in \cite{heidarabadi2020robust} also uses linearized AC equations, and is solved iteratively by Benders Decomposition (BD). In \cite{lorca,yang2021robust}, the multi-period formulation is solved by column and constraint generation and BD, respectively. Effectively, the formulations of \cite{attarha2018adaptive,heidarabadi2020robust,lorca,yang2021robust} can handle the uncertain equality constraints but with relatively high computation time and convergence issues.

\subsection{Research gap}
A qualitative comparison of the attributes of formulations in \cite{bai,louca,Lee2021,attarha2018adaptive,lorca,heidarabadi2020robust,yang2021robust} is given in Table \ref{t1}. It is evident from Table \ref{t1} and Section \ref{sec1b} that the majority of previous formulations, except \cite{bai}, are multi-stage formulations that need to be solved recursively (repetitively solving the same formulation) or iteratively (CP methods, such as BD). The formulation in \cite{bai}, although being non-recursive, does not consider the impact of uncertainty on all variables, except active power generation, needs integer variable representation, and gives highly conservative solutions as the budget of uncertainty constraints are not considered. Unfortunately, such formulations with recursive solution process or integer variable representation are computationally burdensome and may not be suitable for the short-term (5-15 min) operational problems, such as OPF and optimal transmission switching, of large-scale practical AC power networks \cite{posoco}.
\begin{table}[htbp]
\centering
\caption{Attributes of a few previous robust AC OPF formulations and proposed formulation (SSNR - Single Stage, Non-Recursive)}
\begin{tabular}{cccccc}
\hline\hline
\multirow{2}{*}{Ref.} & Robust & \multirow{2}{*}{SSNR} & Optimization & Non- & Uncertainty\\
& Formulation & & Model & integer & Budget \\\hline
\cite{bai} & \textcolor{ForestGreen}{RC} & $\checkmark$ & \textcolor{ForestGreen}{Convex} & $\times$ & $\times$ \\ 
\cite{louca} &  \textcolor{ForestGreen}{RC} & $\times$ & \textcolor{ForestGreen}{Convex} &  $\checkmark$ & $\times$\\
\cite{Lee2021} & \textcolor{ForestGreen}{RC} & $\times$  &  \textcolor{ForestGreen}{Convex} & $\checkmark$ & $\times$\\ 
\cite{attarha2018adaptive}$^*$  & \textcolor{Red}{CP} &  $\times$ &  \textcolor{Red}{LP$^{1,2}$ ,NLP$^3$} & $\times$ & $\checkmark$\\ 
\cite{lorca} & \textcolor{Red}{CP} & $\times$ & \textcolor{ForestGreen}{Convex} & $\checkmark$ & $\checkmark$\\ 
\cite{heidarabadi2020robust}  & \textcolor{Red}{CP}& $\times$ & \textcolor{Red}{LP} & $\checkmark$ & $\checkmark$\\ 
\cite{yang2021robust} & \textcolor{Red}{CP} & $\times$ & \textcolor{ForestGreen}{Convex}& $\times$ & $\checkmark$\\
Prop. & \textcolor{ForestGreen}{RC} & $\checkmark$ & \textcolor{ForestGreen}{Convex}&  $\checkmark$ & $\checkmark$\\\hline\hline\\
\end{tabular}
\label{t1}

$^*$ In \cite{attarha2018adaptive}, a tri-stage formulation is proposed, where some stages are linear while other is non-linear. The superscripts indicate the stages of formulation.
\end{table}
 
\subsection{Contribution and paper organization}
Hence, \emph{a single-stage, non-recursive, time-efficient, dual, robust OPF formulation is proposed in this paper for the short-term operation of AC power network with injection uncertainties. The proposed formulation does not require integer variables to consider budget of uncertainty constraints}. The formulation provides conventional generators' dispatch setpoints, i.e., base-point active power generations and terminal voltage magnitudes, which are immune to short-term injection uncertainties. \emph{Also, in contrast to previous single-stage, non-recursive approach in \cite{bai}, the non-linear impact of uncertainties on variables (such as active loss, reactive generations, and load bus voltages), except the variables related to conventional generators' dispatch setpoints, is inherently modeled without using any affine policy}. As stated earlier, the worst-case realization of uncertainties is an arduous task for actual non-convex ACOPF. This issue is addressed by means of Second-Order Cone (SOC) relaxations that are further tightened by a set of linearized arctangent constraints. Using the property of dual of convexified ACOPF, a novel linear constraint is introduced that helps in the worst-case uncertainty realization. SOC relaxation for mesh network is exact for assumptions stated in \cite{probrobust,hijazi}. Hence, MCS based participation factored AC power flows \cite{plf1} numerically verify the robustness and feasibility of obtained setpoints. Thus, the key attributes of the proposed formulation are
\begin{itemize}
\item A novel single-stage, non-recursive, RC of ACOPF, which utilizes the properties of dual transformation, is proposed to handle short-term uncertainties in loads and RESs, and provides time-efficient robust conventional generators's setpoints, i.e., associated base-point active power generations and terminal voltage magnitudes, immune to all realizations of injection uncertainties.
\item The non-linear impact of uncertainties on variables, except conventional generators' setpoints, is inherently modeled without any affine policy by considering ACOPF constraints for base-case (no uncertainties) and worst-case uncertainties, along with ramp constraints on active power generations of conventional generators.
\item The budget of uncertainty constraints are incorporated for a trade-off between conservatism and robustness of the solution, while avoiding integer variables.
\end{itemize}
Mathematical proofs, extensive numeric results, and comparison with deterministic ACOPF, \cite{bai,louca} reveal that the proposed formulation can be efficiently solved to obtain optimal, robust and feasible conventional generators' setpoints for all uncertainty realizations. Deterministic ACOPF with associated convex formulation are discussed in Section \ref{sec2}. Section \ref{sec3} discusses the proposed RC formulation of ACOPF. Associated numeric results are discussed in Section \ref{sec4} while Section \ref{sec5} concludes the paper and gives directions for future work.

\section{Deterministic ACOPF and Associated Convex Formulation}\label{sec2}
Let $\mathbb{R}$ be the set of real numbers. $\mathcal{N}$, $\mathcal{B}$, $\mathcal{G}$, $\mathcal{L}$ and $\mathcal{R}(\subseteq{\mathcal{G}\cup\mathcal{L}})$ denote the set of buses, branches, generators, load and RES connected buses, respectively. $\mathcal{M}(i)$ is the set of buses connected to bus $i$. $P_{gi}$ ($Q_{gi}$) is the active (reactive) power output of generator $i \in \mathcal{G}$. $V_i\angle\theta_i$ is voltage phasor of bus $i\in\mathcal{N}$. $\theta_{ij}=\theta_i-\theta_j$ and $P_{ij}$ is the active power flow on line $(i,j)\in \mathcal{B}$ from bus $i$ to bus $j$. $P_{di}$ ($Q_{di}$), $P_{ri}$ ($Q_{ri}$) are the nominal active (reactive) load and RES generation at bus $i$, respectively. $\overline{(.)}$ and $\underline{(.)}$ denote the maximum and minimum bounds. $P_{gi}$ and $Q_{gi}$ are undefined for $i\notin\mathcal{G}$. $P_{di}$ and $Q_{di}$ ($P_{ri}$ and $Q_{ri}$) are undefined for $i\notin\mathcal{L}$ ($i\notin\mathcal{R}$). Similar definitions of injections are true for case with uncertainties.

The base-case (deterministic) ACOPF can be stated as
\begin{subequations}\label{model1}
\begin{align}\label{1}
&\min\sum_{i\in\mathcal{G}}\left[a_iP_{gi}+b_i\right]\\\label{2}
&\text{s.t. }P_{gi} - P_i=P_{di}-P_{ri}, \forall i\in\mathcal{N}\\\label{221}
&Q_{gi}- Q_i= Q_{di}-Q_{ri}, \forall i\in\mathcal{N}\\\label{811}
&G_{ii}c_{ii}+\sum_{j\in M(i)}\left[G_{ij}c_{ij}+B_{ij}s_{ij}\right]- P_i=0\,\forall i\in\mathcal{N}\\\label{821}
-&B_{ii}c_{ii}+\sum_{j\in M(i)}\left[G_{ij}s_{ij}-B_{ij}c_{ij}\right]-Q_i=0,\forall i\in\mathcal{N}\\\label{831}
&c_{ii}g_{ij}/|t_{ij}|^2+c_{ij}G_{ij}+s_{ij}B_{ij}-P_{ij}=0,\forall (i,j)\in\mathcal{B}\\\label{Pjieqconv}
&c_{jj}g_{ij}+c_{ij}G_{ji}-s_{ij}B_{ji}-P_{ji}=0,\forall (i,j)\in\mathcal{B}\\\label{16}
&c_{ij}^2+s_{ij}^2=c_{ii}c_{jj},\forall (i,j)\in\mathcal{B}\\\label{17}
&\tan\theta_{ij}=s_{ij}/c_{ij},\forall (i,j)\in\mathcal{B}\\\label{4}
&\underline{P}_{ij}\leq P_{ij}\leq\overline{P}_{ij}, \underline{P}_{ij}\leq P_{ji}\leq\overline{P}_{ij},\forall (i,j)\in\mathcal{B}\\\label{Prlimit}
&\underline{Q}_{ri} \leq Q_{ri} \leq \overline{Q}_{ri},\forall i\in\mathcal{R}\\\label{6}
&\underline{P}_{gi} \leq P_{gi} \leq \overline{P}_{gi}, \underline{Q}_{gi} \leq Q_{gi} \leq \overline{Q}_{gi},\forall i\in\mathcal{G}\\\label{2210}
&\underline{V}_i^2 \leq c_{ii} \leq \overline{V}_i^2 ,\underline{\theta}_i \leq \theta_i \leq \overline{\theta}_i,\forall i\in\mathcal{N},\theta_{ref}=0\\\label{2220}
&\underline{\theta}_{ij} \leq \theta_{ij} \leq \overline{\theta}_{ij},\forall (i,j)\in\mathcal{B}
\end{align}
\end{subequations}
where $a_i$ and $b_i$ are positive cost coefficients of $i^{th}$ generator. Linear cost characteristics are considered as the focus is on the short-term operation of AC power network \cite{posoco,louca}. The nodal power balance equations are \eqref{2}, \eqref{221}. The active and reactive injections ($P_i$ and $Q_i$) at each bus are, respectively, \eqref{811} and \eqref{821}, where $c_{ii}= V_i^2,\forall i\in\mathcal{N}$, and $c_{ij}=V_iV_j\cos\theta_{ij}$, $s_{ij}=V_iV_j\sin\theta_{ij},\forall (i,j) \in\mathcal{B}$. $G_{ij}$ and $B_{ij}$ are real and imaginary parts of $(i,j)^{th}$ element of the bus admittance matrix, respectively. \eqref{831} - \eqref{Pjieqconv} are the active power flows ($P_{ij}$ and $P_{ji}$) on each branch $(i,j)\in\mathcal{B}$, where $g_{ij}$ is the real part of the associated branch impedance and $t_{ij}$ is the complex transformer tap. \eqref{16} and \eqref{17} are valid as per the definitions of $c_{ii}$, $c_{ij}$, and $s_{ij}$. $P_{ij},P_{ji},\forall(i,j)\in\mathcal{B}$ are bounded in \eqref{4}, where $\underline{P}_{ij}=-\overline{P}_{ij}$. The $i^{th}$ RES reactive power output is bounded in \eqref{Prlimit}, where $\overline{Q_{ri}}=-\underline{Q_{ri}}=\sqrt{(S_{ri}^{max})^2-P_{ri}^2}$ and $S_{ri}^{max}$ is the maximum rating of $i^{th}$ RES generator equipped with an inverter \cite{louca}. The conventional generators power outputs are bounded in \eqref{6}. Bus voltage magnitudes and associated phase angles are bounded in \eqref{2210}, where $\theta_{ref}$ is the reference bus voltage phase angle, and $\overline{\theta}_i=\pi/2=-\underline{\theta}_i,\forall i\in\mathcal{N}$. The branch power flows are further constrained by bounds on $\theta_{ij},\forall(i,j)\in\mathcal{B}$ in \eqref{2220} (details on page 17 of \cite{coffrin2014nesta}), where $\underline{\theta}_{ij}=-\overline{\theta}_{ij}$. Limits on branch apparent flow, instead of active power flow, can be considered to limit the branch current magnitude. However, this is beyond the scope of present paper, and for simplicity, constraints in \eqref{4} are considered.

The formulation in \eqref{model1} is non-convex due to the non-convexity of \eqref{16} and \eqref{17}. With injection uncertainties, the non-convexity of ACOPF poses difficulty in quick realization of the associated worst-case, and thus, obtaining a tractable robust solution is arduous \cite{ben2002robust}. Hence, in the associated SOC relaxation, \eqref{16} is convexified as rotated cones \cite{yasasvi2021robust} as
\begin{subequations}
\begin{align}\label{conveq16}
&C_{ij}=2c_{ij},S_{ij}=2s_{ij},\forall(i,j)\in\mathcal{B}\\\label{conveq19} &E_{ij}=c_{ii}-c_{jj},D_{ij}=c_{ii}+c_{jj},\forall(i,j)\in\mathcal{B}\\\label{convcone}
&C_{ij}^2+S_{ij}^2+E_{ij}^2-D_{ij}^2\leq 0,\forall(i,j)\in\mathcal{B}
\end{align}
\eqref{17} can be represented through appropriate constraints using the properties of AC power network. Usually, for typical operating condition of the power network, bus voltage magnitudes are close to $1$ pu and voltage phase angle differences across lines rarely exceed $45^{\circ}$ in \eqref{2220}, as \cite{6981994} reports these differences to not exceed $30^{\circ}$, with most differences being less than $10^{\circ}$. Hence, under such conditions, $c_{ij}\simeq1,\forall(i,j)\in\mathcal{B}$, and
\begin{equation}\label{spll1}
s_{ij}\simeq(1+ \Delta V_i)(1+ \Delta V_j)\theta_{ij}\simeq \theta_{ij} \pm\epsilon_\theta,\forall(i,j)\in\mathcal{B}
\end{equation}
Consequently, from \eqref{spll1}, \eqref{17} can be expressed as \cite{yasasvi2021robust}
\begin{equation}\label{23}
-\epsilon_\theta \leq \theta_{ij}-s_{ij} \leq \epsilon_\theta
\end{equation}
\end{subequations}
where $\epsilon_\theta$ is the allowable error, chosen as per \cite{yasasvi2021robust}. For the same SOC relaxation, i.e., \eqref{conveq16} - \eqref{convcone} and \eqref{23}, the results for Case 2 in Section IV.A of \cite{yasasvi2021robust} strongly indicate the existence of zero optimality gap with same solution between actual non-linear, non-convex ACOPF and associated convex relaxed ACOPF with fixed line plans. Thus, inspired by results of \cite{yasasvi2021robust}, the convexified compact form of \eqref{model1} is (elements in $< >$ are associated Lagrange multipliers)
\begin{subequations}\label{model2}
\begin{align}\label{m2obj}
&\min_{\mathbf{x},\mathbf{y},\mathbf{y}_c}\mathbf{c}^T\mathbf{x}+d\\\label{m2eq}
\text{s.t. }&\mathbf{P}\mathbf{x}+\mathbf{Q}\mathbf{y}=\mathbf{f} & <\mathbf{z}_2>\\\label{m2ineq}
&\mathbf{M}\mathbf{x}+\mathbf{N}\mathbf{y}\leq\mathbf{e} & <\mathbf{z}_1>\\\label{m2cone1}
&\mathbf{y}_c+\mathbf{V}\mathbf{x}+\mathbf{W}\mathbf{y}=\mathbf{0} & <\mathbf{z}_3>\\\label{m2cone2}
&\mathbf{y}_{cij}^T\mathbf{H}\mathbf{y}_{cij}\leq\mathbf{0},\forall(i,j)\in\mathcal{B} & <z_{cij}>
\end{align}
\end{subequations}
where $\mathbf{x}=\{P_{gi},c_{ii},\forall i\in\mathcal{G}\}$. $\mathbf{y}_c=\{\mathbf{y}_{cij},\forall(i,j)\in\mathcal{B}\}$ is the vector of variables used to represent the rotated cones in \eqref{convcone} while $\mathbf{y}$ is the vector of other variables. Hence, \eqref{m2obj} refers to \eqref{1} with $d=\sum_{i\in\mathcal{G}}b_i$. \eqref{m2eq} denotes to \eqref{2} - \eqref{Pjieqconv} while \eqref{m2ineq} denotes \eqref{4} - \eqref{2220} and \eqref{23}. \eqref{m2cone1} represents \eqref{conveq16} and \eqref{conveq19} while \eqref{m2cone2} refers to \eqref{convcone}, where $\mathbf{y}_{cij}=\{C_{ij},S_{ij},E_{ij},D_{ij}\}$. $\mathbf{H}$ in \eqref{m2cone2} is a $4\times4$ diagonal matrix with the first three diagonal elements as $1$ and last diagonal element as $-1$. $\mathbf{0}$ is a zero matrix of appropriate dimension. Based on \eqref{model2}, the proposed robust ACOPF, discussed next, gives robust $\mathbf{x}$ for short-term operation of AC power network, whose feasibility is discussed in Section \ref{robust}. The results in Section \ref{sec4444} also indicate this feasibility for short-term injection uncertainties, which is further validated by MCS based participation factored AC power flows.

\section{Proposed Methodology}\label{sec3}
The formulation in \eqref{model1} can provide base-point active power generations and terminal voltage magnitudes of conventional generators for ACOPF with no uncertainties. However, with injection uncertainties, the same settings are not optimal and result in line overloads with potential of cascaded outages \cite{bienstock2014chance}. This phenomenon is illustrated in Fig. \ref{fig1} for the NESTA $118$ bus system. Hence, short-term injection uncertainties are taken into account for the secure operation of the AC network.
\begin{figure}[htbp]
\centering
\includegraphics[scale=0.7]{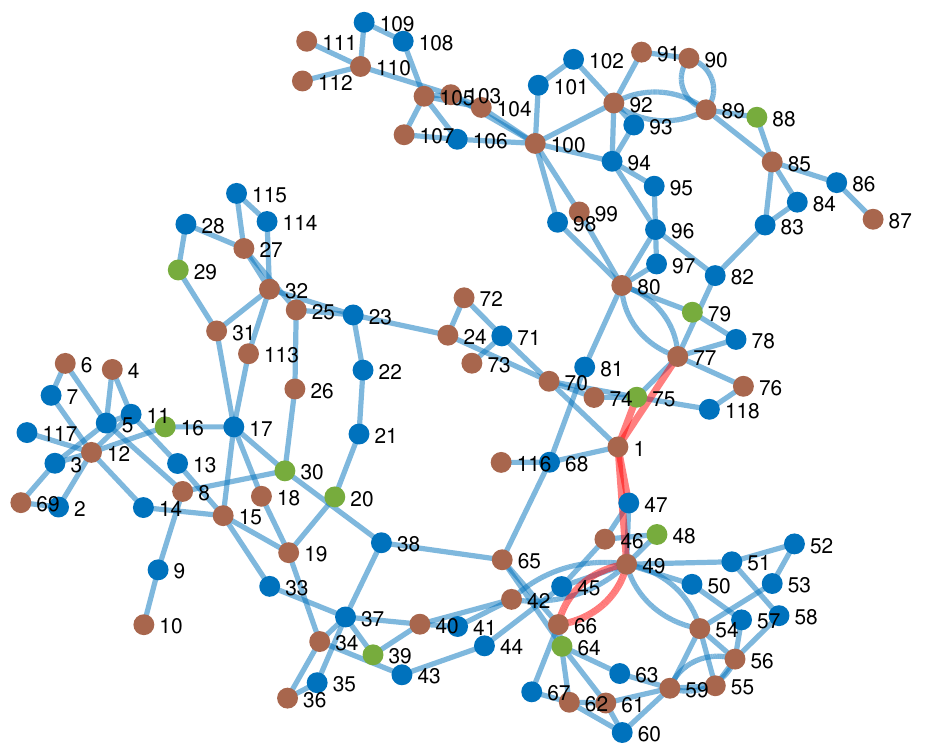}
\caption{Line flow violation due to solution from \eqref{model1} in NESTA $118$ bus system for all RESs and loads injection uncertainties. Conventional generator, RES and load buses are marked by brown, green and blue dots, respectively. RES penetration is $30\%$ of available generation capacities. Deviations in RESs (loads) are $15\%$ ($5\%$) of respective nominal values. Lines in red exceed their flow limits by $3\%-48.6\%$ of respective maximum capacity.
\label{fig1}}
\end{figure}

\subsection{Modeling uncertainties and associated response}\label{sec2a}
\subsubsection{Uncertainty model}\label{sec2b1}
The uncertain injections of active power load and RES generation at bus $i$ are modeled as \cite{8017474}
\begin{subequations}
\begin{align}\label{Pd}
\hat{P}_{di}&=P_{di}+\mu_{di},\forall i\in\mathcal{L}\\\label{Pr}
\hat{P}_{ri}&=P_{ri}+\mu_{ri},\forall i\in\mathcal{R}
\end{align}
where $P_{di}$ ($P_{ri}$) is the forecasted nominal real power load (RES generation). $\mu_{di}$ ($\mu_{ri}$) is the associated deviation in the real power load (RES generation) with zero mean around the respective forecasted value \cite{8017474}. Further, $\bm{\mu}\in\mathcal{U}$, $\mathcal{U}:=\{\bm{\mu}\in\mathbb{R}^{|\mathcal{L}|+|\mathcal{R}|}:\bm{\mu} \in [-\overline{\bm{\mu}},\overline{\bm{\mu}}]\}$, where $\overline{\bm{\mu}}$ is the maximum deviation from the respective nominal load and RES generation (RO requires range of uncertain parameters \cite{robustproof}). Assuming constant power factor operation of loads, the uncertainty in $Q_{di}$ can be stated as
\begin{equation}\label{Qd}
\hat{Q}_{di}=Q_{di}+LR_i\mu_{di},\forall i\in\mathcal{L}
\end{equation}
where $LR_i=\tan\phi_i,\forall i\in\mathcal{L}$, and $\cos\phi_i$ is the associated constant power factor of $i^{th}$ load.

\subsubsection{RES generator model}
RES generators, equipped with power electronics inverters can inject both active and reactive powers. As per \eqref{Pr} and defined range of uncertainty, the uncertain active power output of $i^{th}$ RES generator, i.e., $\hat{P}_{ri},\forall i\in\mathcal{R}$ is in range of $P_{ri}+[-\overline{\mu}_{ri},\overline{\mu}_{ri}]$. Consequently, instead of \eqref{Prlimit}, the uncertain reactive power output of $i^{th}$ RES generator should satisfy
\begin{equation}\label{qrlim}
\underline{\hat{Q}_{ri}}\leq \hat{Q}_{ri}\leq \overline{\hat{Q}_{ri}},\forall i\in\mathcal{R}
\end{equation}
where $\overline{\hat{Q}_{ri}}=-\underline{\hat{Q}_{ri}} =\sqrt{(S_{ri}^{max})^2-(P_{ri}+\overline{\mu}_{ri})^2}$ \cite{louca}.

\subsubsection{Modeling system response}
The deviations in nominal loads and RESs generations in \eqref{Pd} - \eqref{qrlim} due to injection uncertainties change the network power flows even for the short-term operation. To accommodate such changes, the controllable (conventional) generators economically adjust their setpoints to ensure power balance and maintain the desired voltage profile in the network at $\bm{\mu}\neq0$ \cite{8017474}. Hence, the active power output affinely adjusts as per the respective Automatic Generation Control (AGC) action \cite{plf1}, i.e.,
\begin{equation}\label{Ag1}
\hat{P}_{gi}=P_{gi}+\rho_i\psi,\forall i\in\mathcal{G}
\end{equation}
where $\hat{P}_{gi}$ is the active power output of $i^{th}$ generator at $\bm{\mu}\neq0$. $P_{gi}$ and $\rho_i$ are the associated base-point output at $\bm{\mu}=0$ and participation factor, respectively, with $\sum_{i\in\mathcal{G}}\rho_i=1$ \cite{8017474}. $\psi$ is the unknown real-time active power mismatch, due to non-zero uncertain injections and resultant change in network's active power loss with respect to $\bm{\mu}=0$. Hence, as per \eqref{Pd}, \eqref{Pr}, and \eqref{Ag1}, \eqref{2} for $\bm{\mu}\neq0$ is
\begin{equation}\label{pgnew}
P_{gi}+\rho_i\psi-\mu_{di}+\mu_{ri}-\hat{P}_{i}=P_{di}-P_{ri},\forall i\in\mathcal{N}
\end{equation}
where $\hat{P}_i$ is the active power injection at bus $i$ for $\bm{\mu}\neq0$.

Similarly, \eqref{221} for $\bm{\mu}\neq0$, as per \eqref{Qd} and \eqref{qrlim}, is
\begin{equation}\label{qgnew}
\hat{Q}_{gi}-LR_i\mu_{di}+\hat{Q}_{ri}-\hat{Q}_{i}=Q_{di},\forall i\in\mathcal{N}
\end{equation}
where $\hat{Q}_{gi},\forall i\in\mathcal{G}$ and $\hat{Q}_i,\forall i\in\mathcal{N}$ are the respective reactive generation and injection for $\bm{\mu}\neq0$. Also, \eqref{6} for $\bm{\mu}\neq0$ is
\begin{equation}\label{pgnew1}
\underline{P}_{gi}\leq P_{gi}+\rho_i\psi\leq\overline{P}_{gi}, ~ \underline{Q}_{gi}\leq\hat{Q}_{gi}\leq\overline{Q}_{gi},\forall i\in\mathcal{G}
\end{equation}

As per \cite{spyros18}, the conventional generators regulate their reactive power outputs so that $c_{ii},\forall i\in\mathcal{G}$ remain unchanged for $\bm{\mu}\neq0$ and $\bm{\mu}=0$, due to the infrequent action of automatic voltage control as compared to AGC in conventional generators \cite{avr}. Hence, keeping $c_{ii},\forall i\in\mathcal{G}$ unchanged is desirable for the short-term operation with injection uncertainties. It can be noted that $P_i$, $Q_i$, $Q_{gi}$, $Q_{ri}$, defined for $\bm{\mu}=0$, respectively, change to $\hat{P}_i$, $\hat{Q}_i$, $\hat{Q}_{gi}$, $\hat{Q}_{ri}$ for $\bm{\mu}\neq0$ in \eqref{qrlim}, \eqref{pgnew} - \eqref{pgnew1}. Similarly, all variables in \eqref{model2}, except $\mathbf{x}$, i.e., $\mathbf{y}$ and $\mathbf{y}_c$, respectively, non-linearly vary to $\mathbf{\hat{y}}$ and $\mathbf{\hat{y}}_c$ for $\bm{\mu}\neq0$. Additionally, the change in generators' active power in \eqref{Ag1}, i.e., $\rho_i\psi,\forall i\in\mathcal{G}$ is bounded by active power ramp limits as
\begin{equation}\label{ramp1}
-\overline{r}_i\leq\rho_i\psi\leq\overline{r}_i,\forall i\in\mathcal{G}
\end{equation}
\end{subequations}
where $\overline{r}_i\geq0$ is the respective limit.\vspace{-0.5em}

\subsection{Proposed robust ACOPF formulation}
Thus, as per the above discussion, \eqref{m2eq} - \eqref{m2cone2}, respectively, transform to \eqref{m5eq} - \eqref{m5cone2} for $\bm{\mu}\neq0$. Hence, the proposed primal robust ACOPF, which minimizes the total generation cost for the worst-case uncertainty and forecasted operating point while satisfying network constraints for base-case (no uncertainty) and all non-zero realizations of uncertainty, is
\begin{subequations}\label{model5}
\begin{align}\label{m5obj}
&\max_{\bm{\mu}}\min_{\mathbf{x},\mathbf{y},\mathbf{y}_c,\mathbf{\hat{y}},\mathbf{\hat{y}}_c,\psi}\mathbf{c}^T\mathbf{x}+d\text{ s.t. }\eqref{m2eq} - \eqref{m2cone2}\\\label{m5eq}
&\mathbf{P}\mathbf{x}+\mathbf{Q}\mathbf{\hat{y}}+\mathbf{h}_E\psi+\mathbf{K}_E\bm{\mu}=\mathbf{f} & <\mathbf{\hat{z}}_2>\\\label{m5ineq}
&\mathbf{M}\mathbf{x}+\mathbf{N}\mathbf{\hat{y}}+\mathbf{h}_I\psi\leq\mathbf{\hat{e}} & <\mathbf{\hat{z}}_1>\\\label{m5cone1}
&\mathbf{\hat{y}}_c+\mathbf{V}\mathbf{x}+\mathbf{W}\mathbf{\hat{y}}=\mathbf{0} & <\mathbf{\hat{z}}_3>\\\label{m5cone2}
&\mathbf{\hat{y}}_{cij}^T\mathbf{H}\mathbf{\hat{y}}_{cij}\leq\mathbf{0},\forall(i,j)\in\mathcal{B} & <\hat{z}_{cij}>\\\label{m5ramp}
&\mathbf{h}_r\psi\leq\mathbf{r} & <\mathbf{\hat{z}}_r>
\end{align}
\end{subequations}
where $\mathbf{c}$, $d$, $\mathbf{P}$, $\mathbf{Q}$, $\mathbf{f}$, $\mathbf{M}$, $\mathbf{N}$, $\mathbf{V}$, $\mathbf{W}$, and $\mathbf{H}$ are as per \eqref{model2}. $\mathbf{e}$ in \eqref{m2ineq} modifies to a known $\mathbf{\hat{e}}$ in \eqref{m5ineq} due to \eqref{qrlim} replacing \eqref{Prlimit} for $\bm{\mu}\neq0$ while bounds of other constraints in \eqref{m2ineq} remain unchanged, similar to the case in \eqref{6} and \eqref{pgnew1}. $\mathbf{h}_E$ and $\mathbf{h}_I$ are non-zero column vectors corresponding to \eqref{pgnew} and \eqref{pgnew1}, respectively. Similarly, $\mathbf{K}_E$ is a coefficient matrix corresponding to $\bm{\mu}\neq0$ in \eqref{pgnew}, \eqref{qgnew}. \eqref{m5ramp} corresponds to \eqref{ramp1}. The following remarks present three critical aspects of \eqref{model5}.
\begin{remark}
\textit{For a known $\bm{\mu}\neq0$, it can be proven that without \eqref{m2eq} - \eqref{m2cone2}, there exists a generator $i\in\mathcal{G}$ for which \eqref{model5} gives $\rho_i\psi=\overline{r}_i$ and $P_{gi}\in[\underline{P}_{gi},\overline{P}_{gi}]-\overline{r}_i$, where it is possible that $\underline{P}_{gi}\leq\overline{r}_i$ as $\underline{P}_{gi}\geq0$ and $0\leq\overline{r}_i\leq\overline{P}_{gi}$ (non-negative conventional active generation and ramp limits). Consequently, \eqref{model5} gives $P_{gi}\leq0$ as optimal solution for minimum cost, which is impractical. With \eqref{m2eq} - \eqref{m2cone2} part of \eqref{model5}, this does not exist, due to \eqref{m5ramp} and $P_{gi}\in[\max(\underline{P}_{gi},\underline{P}_{gi}-\rho_i\psi),\min(\overline{P}_{gi},\overline{P}_{gi}-\rho_i\psi)]\forall i\in\mathcal{G}$}.
\end{remark}
\begin{remark}
\textit{In \eqref{model5}, $\mathbf{y}$ and $\mathbf{y}_c$ non-linearly adjust to $\mathbf{\hat{y}}$ and $\mathbf{\hat{y}}_c$ for $\bm{\mu}\neq0$, respectively, without any affine control policy}.
\end{remark}
\begin{remark}\label{remark3}
\textit{For $\bm{\mu}=0$, the resultant change in network's active power loss with respect to base-case (no uncertainty) is $0$. Hence, $\psi=0$ for $\bm{\mu}=0$. Consequently, $\hat{P}_{gi}={P}_{gi}$ and $P_i=\hat{P}_i$ in \eqref{Ag1}, \eqref{pgnew}, respectively. In fact, \eqref{m5eq} - \eqref{m5cone2} are, respectively, equivalent to \eqref{m2eq} - \eqref{m2cone2} for $\bm{\mu}=0$. Nonetheless, \eqref{m5ramp} couples the conditions $\bm{\mu}=0$ and $\bm{\mu}\neq0$ in \eqref{model5}. Since \eqref{m5ramp} regulates the real-time change in active generation of conventional generators, the level of short-term injection uncertainties that can be handled depends on ramp limits}.
\end{remark}

It is evident from \eqref{m5eq} that uncertain injections appear in equality constraints. Uncertain terms in equality constraints restrict the primal feasible region \cite{gorissen2015practical}. However, in the dual form, the uncertain terms appear only in the dual objective with dual constraints free of uncertain injections. Hence, a novel single-stage non-recursive convex RC of \eqref{model5} is proposed next, which provides computationally time-efficient robust setpoints of conventional generators for the short-term operation of AC power network with injection uncertainties.

\subsection{Novel single-stage non-recursive dual RC}\label{robust}
The dual of \eqref{model5} aids in the worst-case realization of $\bm{\mu}\in\mathcal{U}$, as $\bm{\mu}$ appears only in dual objective to be maximized. Further, the dual of \eqref{model5} from conic dual theory \cite{boyd2004convex} does not guarantee strong duality. Hence, a novel dual transformation is utilized using the \emph{Karush-Kuhn-Tucker} conditions of the associated Lagrange that guarantees strong duality between dual of \eqref{model5} and \eqref{model5}. This novel dual transformation of \eqref{model5} for $\bm{\mu}\in\mathcal{U}$ is
\begin{subequations}\label{m51}
\begin{align}\label{m5}
&\max d-\mathbf{f}^T(\mathbf{z}_2+\mathbf{\hat{z}}_2)-\mathbf{e}^T\mathbf{z}_1-\mathbf{\hat{e}}^T\mathbf{\hat{z}}_1-\mathbf{r}^T\mathbf{\hat{z}}_r+\bm{\mu}^T\mathbf{K}_E^T\mathbf{\hat{z}}_2\\\label{uncs}
&\text{s.t. }\mathbf{Q}^T\mathbf{z}_2+\mathbf{N}^T\mathbf{z}_1+\mathbf{W}^T\mathbf{z}_3=\mathbf{0}\; <\mathbf{y}>\\
&\mathbf{Q}^T\mathbf{\hat{z}}_2+\mathbf{N}^T\mathbf{\hat{z}}_1+\mathbf{W}^T\mathbf{\hat{z}}_3=\mathbf{0}\; <\mathbf{\hat{y}}>\\\label{m5eqx}
&\mathbf{P}^T(\mathbf{z}_2+\mathbf{\hat{z}}_2)+\mathbf{M}^T(\mathbf{z}_1+\mathbf{\hat{z}}_1)+\mathbf{V}^T(\mathbf{z}_3+\mathbf{\hat{z}}_3)=-\mathbf{c}\;<\mathbf{x}>\\
&\mathbf{h}_E^T\mathbf{\hat{z}}_2+\mathbf{h}_I^T\mathbf{\hat{z}}_1+\mathbf{h}_r^T\mathbf{\hat{z}}_r=0\;<\psi>\\\label{dcone}
&\mathbf{z}_{3ij}^T\mathbf{H}\mathbf{z}_{3ij}\leq\mathbf{0},\mathbf{\hat{z}}_{3ij}^T\mathbf{H}\mathbf{\hat{z}}_{3ij}\leq\mathbf{0},\forall (i,j)\in\mathcal{B}\\\label{unce34}
&\mathbf{z}_{3ij}+2z_{cij}\mathbf{H}\mathbf{y}_{cij}=0,\forall (i,j)\in\mathcal{B}\;<\mathbf{y}_{cij}>\\\label{unce35}
&\mathbf{\hat{z}}_{3ij}+2\hat{z}_{cij}\mathbf{H}\mathbf{\hat{y}}_{cij}=0,\forall (i,j)\in\mathcal{B}\;<\mathbf{\hat{y}}_{cij}>\\\label{unce1}
&\mathbf{z}_1,\mathbf{z}_{c},\mathbf{\hat{z}}_1,\mathbf{\hat{z}}_{c},\mathbf{\hat{z}}_r\geq\mathbf{0},\mathbf{z}_2,\mathbf{z}_3,\mathbf{\hat{z}}_2,\mathbf{\hat{z}}_3\text{ unrestricted}
\end{align} 
\end{subequations}
where $\mathbf{z}_3=\{\mathbf{z}_{3ij},\forall(i,j)\in\mathcal{B}\}$, $\mathbf{\hat{z}}_3=\{\mathbf{\hat{z}}_{3ij},\forall(i,j)\in\mathcal{B}\}$, $\mathbf{z}_c=\{z_{cij},\forall(i,j)\in\mathcal{B}\}$ and $\mathbf{\hat{z}}_c=\{\hat{z}_{cij},\forall(i,j)\in\mathcal{B}\}$. \eqref{m5} - \eqref{dcone} are obtained from conic duality theory \cite{boyd2004convex}, whereas \eqref{unce34} - \eqref{unce35} are normalization constraints to ensure strong duality, similar to extended self dual model in \cite{sedumi11,sedumi}. More details on derivation of \eqref{m51} and proof of strong duality are in Appendix A of \cite{yasasvi2021robust}. \emph{Since strong duality exists, the Lagrange multipliers in \eqref{m51} are same as the primal variables of \eqref{model5}}. Further, $\bm{\mu}^TK_E^T\mathbf{\hat{z}_2}$ in \eqref{m5} can be re-written as
\begin{equation}\label{repres}
\sum_{j=1}^{|\mathcal{U}|}\{\max(R_j,0)\overline{\mu}_j+\min(R_j,0)\underline{\mu}_j\}
\end{equation}
where $\overline{\mu}_j$ and $\underline{\mu}_j$ are upper and lower bounds of $\mu_j$ in Section \ref{sec2b1}. $|\mathcal{U}|$ is cardinality of $\mathcal{U}$ and $\mathbf{R}=\mathbf{K}_E^T\mathbf{\hat{z}_2}$.

For the worst-case realization of short-term injection uncertainties, it is evident that a positive (non-positive) element of $\mathbf{R}$ in \eqref{m5} should correspond to the associated element of $\bm{\mu}$ being at its maximum (minimum). This is inherently ensured in \eqref{repres}. Therefore, the proposed novel single-stage non-recursive RC of \eqref{m51} can be stated by its epigraph \cite{gorissen2015practical} as
\begin{subequations}\label{m6}
\begin{align}\label{od}
&\max d-\mathbf{f}^T(\mathbf{z}_2+\mathbf{\hat{z}}_2)-\mathbf{e}^T\mathbf{z}_1-\mathbf{\hat{e}}^T\mathbf{\hat{z}}_1-\mathbf{r}^T\mathbf{\hat{z}}_r+\mathbf{1}^T\mathbf{t}\\\label{16b}
&\text{s.t. }R_j^+-R_j^-R_j=0,\forall j=1,\dots,|\mathcal{U}|,\eqref{uncs} - \eqref{unce1}\\\label{ru}
&R_j^+\underline{\mu}_j-R_j^-\overline{\mu}_j\leq t_j\leq R_j^+\overline{\mu}_j-R_j^-\underline{\mu}_j,\forall j=1,\dots,|\mathcal{U}|\\\label{16d}
&0\leq R_j^+,R_j^-\leq T,\forall j=1,\dots,|\mathcal{U}|
\end{align}
\end{subequations}
where $\mathbf{1}$ and $\mathbf{t}$ are column vectors of length $|\mathcal{U}|$. $R_j$ is the $j^{th}$ element of $\mathbf{R}$, and is the coefficient of $\mu_j$, the $j^{th}$ element of $\bm{\mu}$ in \eqref{m5}. $T$ is a large Big-M number. $R_j^+$ and $R_j^-$ take appropriate values as per \eqref{16b} and \eqref{16d}, based on the sign of $R_j$. Depending on $R_j^+$ and $R_j^-$, $t_j$ in \eqref{ru} chooses an appropriate value so that the objective in \eqref{od} is maximized. The proposed RC \eqref{m6} is convex and provides robust conventional generators' setpoints as per the following discussion for all realizations of short-term injection uncertainties.
\begin{remark}
\textit{To the best of authors' knowledge, this is the first single-stage, non-recursive RC formulation of robust ACOPF using duality theory, and is the main contribution of this paper. \eqref{model5} and \eqref{m51} are primal-dual pairs, between which strong duality exists \cite{yasasvi2021robust}, and \eqref{m6} is the RC of \eqref{m51}. Hence, $\mathbf{x}$, obtained as Lagrange multipliers of \eqref{m5eqx}, are actually the desired robust primal variables in \eqref{model5}. Since, the desired robust $\mathbf{x}$ is directly obtained from \eqref{m6}, there is no need of a recursive solution procedure. Hence, \eqref{m6} efficiently provides the robust values of $P_{gi}$, $c_{ii},\forall i\in\mathcal{G}$ in a single-stage non-recursive manner, suitable for the short-term operation of AC power network}.
\end{remark}
\begin{lemma}\label{convfeasibility}
The feasible set of RC in \eqref{m6} is a closed convex set, so that the RC of an uncertain convex program is a convex program. If the RC is feasible, then all instances of uncertain program are feasible, and the robust optimal value is greater or equal to optimal values of all instances. 
\end{lemma}
\begin{proof}
Please see the proof of Proposition 2.2 in \cite{robustproof}.
\end{proof}
\begin{lemma}\label{theorem1}
As per Theorem 1 of \cite{probrobust}, with availability of flexible generation in \eqref{Ag1}, and $\underline{Q}_{gi}=-\infty,\forall i\in\mathcal{G}$, \eqref{m6} is exact for robust $\mathbf{x}$ in a cyclic network with cycle of size 3.
\end{lemma}
\begin{proof}
Please see proof in Appendix \ref{appB}.
\end{proof}
\begin{remark}
\textit{In Lemma \ref{theorem1}, by exact, it is meant that the optimal value of \eqref{m6} is same as the exact actual robust ACOPF model, and also, for any admissible value of $\bm{\mu}$, there exists \textbf{$c_{ii},c_{ij},s_{ij},\theta_i$} for which \eqref{16} and \eqref{17} are satisfied. It is to be noted that the robustness and exactness of proposed formulation is ensured only for robust $\mathbf{x}$. Further, even if the assumptions in Lemma \ref{theorem1} are not true, the feasibility of robust setpoints ($P_{gi}$, $c_{ii},\forall i\in\mathcal{G}$) obtained from \eqref{m6} is numerically ensured via MCS based participation factored AC power flows in Section \ref{sec4444}}.
\end{remark}
The paper focuses on the short-term operation of AC power network, where the deviations in short-term injection uncertainties of loads and RES generations are not high \cite{posoco}. Hence, no infeasibility issues are encountered in the proposed formulation for the considered test cases. \textit{However, if infeasibility occurs for large injection uncertainties (please see Remark \ref{remark3}), load curtailment can be used as additional emergency corrective action \cite{loadshedd1}, which shall require solution of an exclusive corrective security constrained ACOPF for the specific case}. Although, it may be noted that solving corrective security constrained ACOPF for specific infeasibility may be computationally burdensome for short-term operation of power network. The budget of uncertainty constraints to reduce conservatism of the robust setpoints from \eqref{m6} are as follows.

\subsection{Budget of uncertainty constraints}\label{sec33}
With respect to \eqref{m6}, the budget of uncertainty constraints are
\begin{subequations}
\begin{equation}\label{b1}
\mathcal{U}=\{\acute{\mu}_j:\acute{\mu}_j=\alpha_j\mu_j,\mu_j\in[\underline{\mu}_j,\overline{\mu}_j],\alpha_j\in\mathcal{A};\forall j=1,\dots,|\mathcal{U}|\}
\end{equation}
\begin{equation}\label{b2}
\mathcal{A}=\{\alpha_j:\sum_{j=1}^{|\mathcal{U}|}\alpha_j=\Gamma,\alpha_j\in\{0,1\},\forall j=1,\dots,|\mathcal{U}|\}
\end{equation}
where $\Gamma$ is the budget parameter. $\Gamma$ takes integer values in $[0,|\mathcal{U}|]$ and $\alpha_j$ indicates the status of $\mu_j$. The incorporation of \eqref{b1}, \eqref{b2} in \eqref{m6} ensures that only a subset of the uncertain parameters affect the optimal robust solution. \eqref{b1} can be considered in \eqref{m6} by modifying \eqref{ru} as
\begin{equation}\label{bud18}
\alpha_j(R_j^+\underline{\mu}_j-R_j^-\overline{\mu}_j)\leq t_j\leq\alpha_j(R_j^+\overline{\mu}_j-R_j^-\underline{\mu}_j)
\end{equation}
$\forall j=1,\dots,|\mathcal{U}|$. The constraint \eqref{bud18} consists of bilinear terms. With the symmetric uncertainty, i.e., $\underline{\mu_j}= -\overline{\mu}_j$, considered, \eqref{bud18} can be written as $\alpha_j(R_j^++R_j^-)\underline{\mu}_j\leq t_j\leq\alpha_j(R_j^+ + R_j^-)\overline{\mu}_j$, or $ |{t_j}| \leq \alpha_j(R_j^+ + R_j^-)\overline{\mu}_j $  which is a convex constraint. Further, to avoid consideration of the integer variables $\alpha_j$  in \eqref{b2} so that the computational burden of solving the proposed formulation for short-term operation of AC power network is minimized, \eqref{b2} can be restated as
\begin{equation}\label{bud19}
\sum_{j=1}^{|\mathcal{U}|}\alpha_j=\Gamma;\alpha_j\in[0,1];\alpha_j(\alpha_j-1)\leq0
\end{equation}
\end{subequations}
$\forall j=1,\dots,|\mathcal{U}|$. In \eqref{bud19}, $\alpha_j$ need not be an integer variable. Nevertheless \eqref{b2} is automatically satisfied, when \eqref{bud19} is satisfied. Hence, the proposed single-stage non-recursive non-integer dual RC with budget of uncertainty constraints and with respect to \eqref{m6} is
\begin{align}\nonumber
&\max d-\mathbf{f}^T(\mathbf{z}_2+\mathbf{\hat{z}}_2)-\mathbf{e}^T\mathbf{z}_1-\mathbf{\hat{e}}^T\mathbf{\hat{z}}_1-\mathbf{r}^T\mathbf{\hat{z}}_r+\mathbf{1}^T\mathbf{t}\\\label{eq10}
&\text{s.t. }\eqref{16b}, \eqref{16d}, \eqref{bud18}, \eqref{bud19}
\end{align}
\begin{remark}
\textit{\eqref{eq10} is the dual RC of \eqref{model5} with budget of uncertainty constraints, which allows a trade-off between conservatism and robustness of optimal solution}.
\end{remark}
Depending on the specified $\Gamma\in\mathbb{Z}$, the following are true
\begin{itemize}
\item If $\Gamma=0$, then no parameters are uncertain, and \eqref{eq10} is same as \eqref{model2} (due to strong duality stated in Section \ref{robust}).
\item If $\Gamma=|\mathcal{U}|$, then all parameters are uncertain, and \eqref{eq10} is same as \eqref{m6}.
\item $0<\Gamma<|\mathcal{U}|$, then \eqref{eq10} automatically chooses the short-term injections that are supposed to be uncertain, and consequently, a trade off is ensured between robustness and conservatism of the obtained $\mathbf{x}$.
\end{itemize}

\section{Numeric Results}\label{sec4}
The efficacy of the proposed novel dual, convex RC of ACOPF with short-term injection uncertainties and associated formulation with budget of uncertainty constraints are verified on four NESTA \cite{coffrin2014nesta} and the case9241pegase \cite{pegase} systems in terms of generation cost and computation time. Robustness analysis for in-range (interval uncertainty) and out-of-range samples on independent and identically distributed scenarios and feasibility verification of the obtained robust solution is done through MCS based participation factored AC power flow \cite{plf1}. The robust solution from the proposed formulation is further compared with the solution from the deterministic ACOPF and the solutions of formulations reported in \cite{bai,louca}. RESs are assumed to be present in the systems with a maximum penetration capacity of $30\%$. The ramp rate of a conventional generator in \eqref{ramp1} is $\pm0.75$ times the associated base-point output. Since, the emphasis is on ACOPF with short-term injection uncertainties, uncertainties beyond $15\%$ are not considered. All numeric results are obtained in MATLAB using SeDuMi \cite{sedumi} on an Intel i7-2600 CPU with 3.40 GHz processor and 8 GB RAM.

\subsection{Performance of \eqref{m6} and impact of variable participation factors}
\begin{table}[htbp]\vspace{-0.5em}
\centering
\caption{Optimal objective and computation time for $5\%$ load uncertainty and varying uncertainties in RES generations}
\label{t3}
\resizebox{\columnwidth}{!}{%
\begin{tabular}{c|cc|cc|cc|cc|cc}
\hline\hline
$\% $ RES  & \multicolumn{2}{c|}{$\Gamma = |\mathcal{U}|$} & \multicolumn{2}{c|}{$\Gamma = 1$} & \multicolumn{2}{c|}{$\Gamma = 3$} & \multicolumn{2}{c|}{$\Gamma = 5$} & \multicolumn{2}{c}{Violation Probability} \\\cline{2-11}
uncertainty  & Obj. value & Time & Obj. value & Time & Obj. value & Time & Obj. value & Time & in- & out-of-  \\
& ($\$/h$) & (s) & ($\$/h$) & (s) & ($\$/h$) & (s) & ($\$/h$) & (s) & range (\%) & range (\%)  \\\hline
\multicolumn{11}{c}{\textbf{NESTA 14}} \\ \hline
0  & 162.630 & 0.563 & 158.994 & 0.569 & 160.700 & 0.571 & 161.490 & 0.629 & 0 & 5.96 \\
5  & 164.379 & 0.635 & 158.994 & 0.563 & 160.925 & 0.592 & 162.442 & 0.545 & 0 & 12.1\\
10  & 166.138 & 0.565 & 158.994 & 0.594 & 162.472 & 0.543 & 164.193 & 0.577 & 0 & 45.0 \\
15  & 167.906 & 0.561 & 159.500 & 0.568 & 164.225 & 0.552 & 165.954 & 0.566 & 0  & 53.9 \\\hline
\multicolumn{11}{c}{\textbf{NESTA 57}} \\ \hline
0 & 778.488 & 0.839 & 747.790 & 0.780 & 757.769 & 0.716 & 762.010 & 0.7552 & 0 & 34.7\\
5 & 789.461 & 0.853 & 747.790 & 0.780 & 758.609 & 0.725 & 768.603 & 0.724 & 0 & 48.3 \\
10 & 800.599 & 0.767 & 747.790 & 0.741 & 769.592 & 0.704 & 779.599 & 0.690 & 0  & 67.6 \\
15 & 811.900 & 0.789 & 752.904 & 0.722 & 780.739 & 0.723 & 790.760& 0.718 & 0 & 86.1 \\\hline
\multicolumn{11}{c}{\textbf{NESTA 118}} \\ \hline
0  & 2547.294 & 1.913 & 2398.124 & 1.307 & 2410.392 & 1.311 & 2420.288 & 1.284 & 0 & 19.1 \\
5  & 2592.993 & 1.825 & 2398.124 & 1.332 & 2410.392 & 1.398 & 2419.372 & 1.326 & 0 &  27.4 \\
10  & 2635.046 & 1.825 & 2398.124 & 1.387 & 2415.020 & 1.421 & 2431.605 & 1.238 & 0 & 54.2 \\
15  & 2656.603 & 1.843 & 2398.890 & 1.871 & 2424.595  & 1.678 & 2449.642 & 1.236 & 0 & 73.1 \\\hline
\multicolumn{11}{c}{\textbf{NESTA 300}} \\ \hline
0 & 14321.926 & 29.224 & 12605.007 & 13.024 & 12700.644 & 14.180 & 12829.027 & 12.692 & 0 & 49.2 \\
5 & 14402.808 & 21.697 & 12605.007 & 13.024 & 12700.644 & 14.868 & 12900.267 & 13.780 & 0 & 84.6  \\ \hline
\multicolumn{11}{c}{\textbf{PEGASE 9241}} \\ \hline
0 & 284496.88 & 126.741 & 268621.05 & 115.691 & 268691.84 & 128.132 & 268756.49 & 121.285 & 0 & 68.2 \\
5 & 286849.86 & 128.423 & 268635.39 & 119.671 & 268716.11 & 114.568 & 268786.58 & 118.172 & 0 & 94.6  \\ \hline\hline
\end{tabular}\vspace{-0.5em} 
}
\end{table}
The performance of the formulation in \eqref{m6} is tested in the presence of all injection uncertainties from loads and RESs, i.e., $\Gamma=|\mathcal{U}|$. $5\%$ uncertainty is assumed in all system loads. With fixed participation factors, i.e., $\rho_i=a_i^{-1}/(\sum_{i\in\mathcal{G}}a_i^{-1}),\forall i\in\mathcal{G}$, the optimal objective and computation time for different systems are shown in the second and third columns of Table \ref{t3} for different levels of RES injection uncertainties. It can be observed that the optimal objective in second column of Table \ref{t3} only increases by a maximum of $5\%$ of the associated objective costs with no uncertainty in the RES generations. The reported computation times also indicate that the proposed formulation may be used for short-term ACOPF with short-term injection uncertainties.

With generation participation factors as additional variables, the optimal objective can be further reduced. However, with this, \eqref{pgnew} becomes a bilinear constraint. Hence, \eqref{m6} is solved separately to analyze the impact of variable participation factors with $\psi$ and $\bm{\mu}$ fixed for the worst-case uncertainty realization (obtained with fixed $\rho_i$ in \eqref{m6}) and considering $\rho_i,\forall i\in\mathcal{G}$ as new additional optimization variables. With this formulation, for $5\%$ load, the reduction in optimal objective is $0.139\%$ for $5\%$ RES generation uncertainty, and $0.328\%$ for $15\%$ RES generation uncertainty, in NESTA $118$ bus system.

\subsection{Impact of budget of uncertainty in \eqref{eq10}}
A trade off between the robustness and conservatism of the obtained setpoints can be ensured by solving the proposed formulation in \eqref{eq10}. $\Gamma$ with \eqref{bud19} automatically decide the number of uncertain parameters that are allowed to deviate from their respective nominal values, as discussed in Section \ref{sec33}. For different $\Gamma$, the optimal cost and associated computation time for different NESTA systems and case9241pegase system are given in second to ninth columns of Table \ref{t3}. The maximum reduction in optimal cost is $5.00\%$, $7.26\%$ and $9.70\%$, for the NESTA $14$ bus, $57$ bus and $118$ bus systems, respectively, corresponding to $\Gamma=1$ and $15\%$ RES uncertainty. While for $\Gamma=1$ and $5\%$ RES uncertainty, the same reduction is $12\%$ and $6\%$ for NESTA $300$ bus system and case9241pegase system, respectively. Further, the computation time of solving the proposed formulation in \eqref{eq10}, generally, reduces with reduced $\Gamma$ due to lesser number of uncertain injections.

Also, it can be noted from Table \ref{t3} that certain objectives are same for $\Gamma=1$ and $\Gamma=3$, even for different RES uncertainty, as load uncertainties are more dominant. With reduced $\Gamma$, the AC power network's security margin also tends to reduce, as shown in Fig. \ref{fig.57}, because incorporating the budget of uncertainty constraints causes critical lines to carry more power, thus reducing the overload margin. For $5\%$ load and $10\%$ RES uncertainty, the brown line in Fig. \ref{fig.57} is the maximum active power limit of each transmission line in the NESTA $57$ bus system, while the blue and orange lines indicate the active power flows on lines for $\Gamma=5$ and $\Gamma=|\mathcal{U}|$, respectively. These active power line flows are the maximum possible values for each line, obtained from MCS ($N_s=10000$) based participation factored AC power flow with robust generator setpoints from \eqref{eq10} in Algorithm \ref{alg1} for associated realizations of uncertainties. None of the $N_s$ scenarios have power flow divergence or constraint violation, which also indicates the feasibility of obtained robust setpoints for actual ACOPF.
\begin{figure}[htbp]\vspace{-0.5em}
\centering
\includegraphics[width=0.8\textwidth,height=4.4cm]{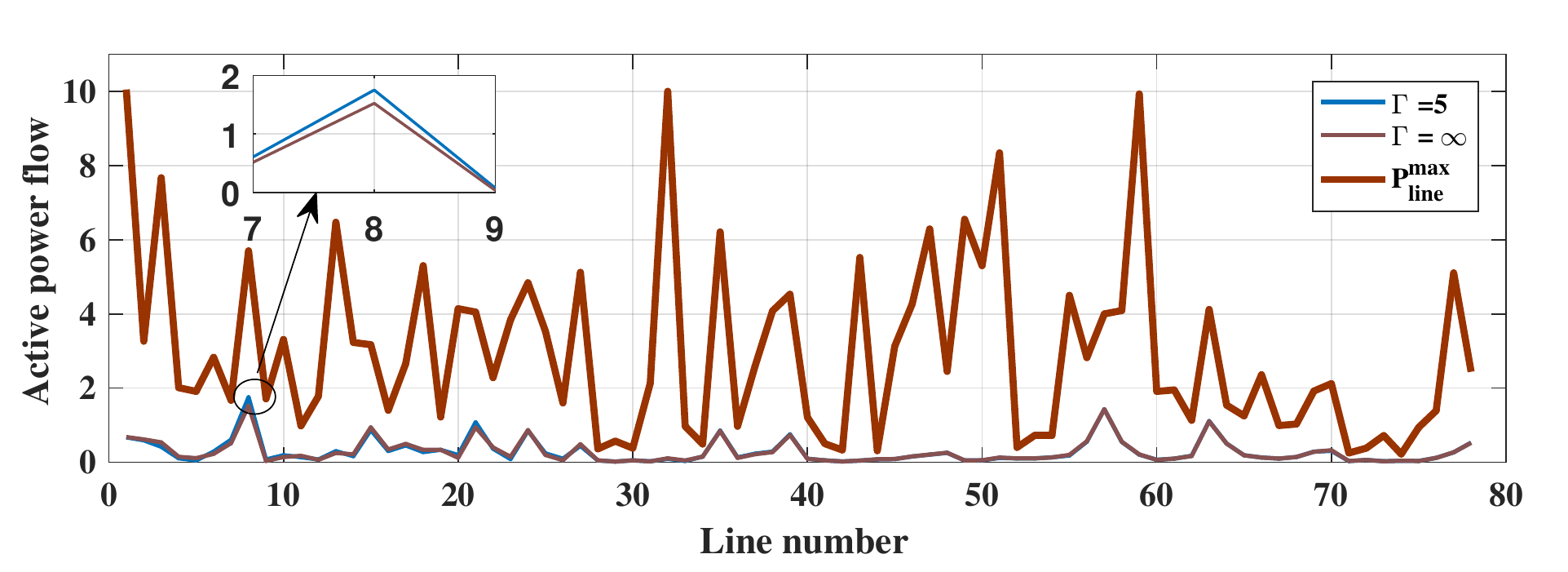}
\caption{Active power flows for NESTA $57$ bus system for different $\Gamma$ at $5\%$ load and $10\%$ RES injections' uncertainties}
\label{fig.57}\vspace{-0.5em}
\end{figure}

\begin{algorithm}[htbp]
\SetAlgoLined
\caption{MCS based robustness check}\label{alg1}
\textbf{Initialization:} Obtain robust setpoints $\mathbf{x}$ from \eqref{eq10}. For the chosen $\mathcal{U}=[\underline{\bm{\mu}},\overline{\bm{\mu}}]$, create $N_s$ random scenarios of uncertain injections. Set $count=1$;\\
\While{$count\neq N_s$}{
run participation factored AC power flow \cite{plf1} using $\mathbf{x}$ for each scenario
\eIf{converged}{Evaluate the inequality constraints' values in exact ACOPF
\eIf{all inequality constraints are satisfied}{go to next sample, {$count\gets count + 1$}}{$\mathbf{x}$ is not robust}}{$\mathbf{x}$ is not robust}}
\textbf{$\mathbf{x}$ is robust}
\end{algorithm}

\subsection{Validating robustness and feasibility of obtained setpoints}\label{sec4444}
The conventional generators' setpoints obtained from \eqref{eq10} or \eqref{m6} are robust against the injection uncertainties, when for all realizations of uncertainties as defined by the associated interval uncertainty sets, the AC power flow converges and provides a solution that satisfies all the physical and operational limits of the AC power network. Also, the proposed formulations use SOC relaxation, which are exact for mesh networks with assumptions in \cite{probrobust,hijazi}. Hence, the feasibility of the obtained robust solution also needs to be validated for exact ACOPF. Further, as the paper focuses on short-term operation of AC power network, validating the robustness and feasibility of the obtained setpoints (for actual non-determistic ACOPF) by running deterministic ACOPFs is not a viable option. Additionally, validating robustness for short-term operation by solving exclusive actual non-linear non-convex ACOPFs may defeat the purpose of obtaining time-efficient robust setpoints. Therefore, the feasibility and robustness of the setpoints from \eqref{m6} and \eqref{eq10} are verified with MCS based participation factored AC power flow \cite{plf1} in Algorithm \ref{alg1} for all cases in Table \ref{t3}. For each scenario of the $N_s=10000$ independent and identically distributed scenarios, the participation factored AC power flow converges and all inequality constraints of exact ACOPF are satisfied at the converged solution for the considered uncertainties. Hence, the tenth column of Table \ref{t3} reports $0\%$ violation for all the in-range scenarios robustness analysis.

Further, the maximum possible values of inequality constraints in exact ACOPF after robustness check with $N_s=10000$ scenarios in Algorithm \ref{alg1} for $5\%$ uncertainty in all system loads and $15\%$ uncertainty in all RES injections in NESTA $118$ bus system are shown in Fig. \ref{fig.51} and Fig. \ref{fig.52}. In Fig. \ref{fig.51}(a), the normalized active power generations and associated ramp are shown. Generators $14$ and $46$ in Fig. \ref{fig.51}(a) refer to buses $31$ and $103$ in Fig. \ref{fig.52}(a), respectively. Actual $\hat{P}_g$, $P_g$ and ramp for bus $31$ are $0.23$pu, $0.217$pu and $0.0131$pu, respectively. For bus $103$, these values are $0.98$pu, $0.9646$pu and $0.0154$pu, respectively. It is clear from Fig. \ref{fig.57} - Fig. \ref{fig.52} that inequality constraints of exact ACOPF are within their limits. Thus, the robust setpoints are feasible for actual ACOPF (as stated in Appendix \ref{appB}) and immune to uncertainties.
\begin{figure}[htbp]\vspace{-0.5em}
\centering
\includegraphics[scale=0.7]{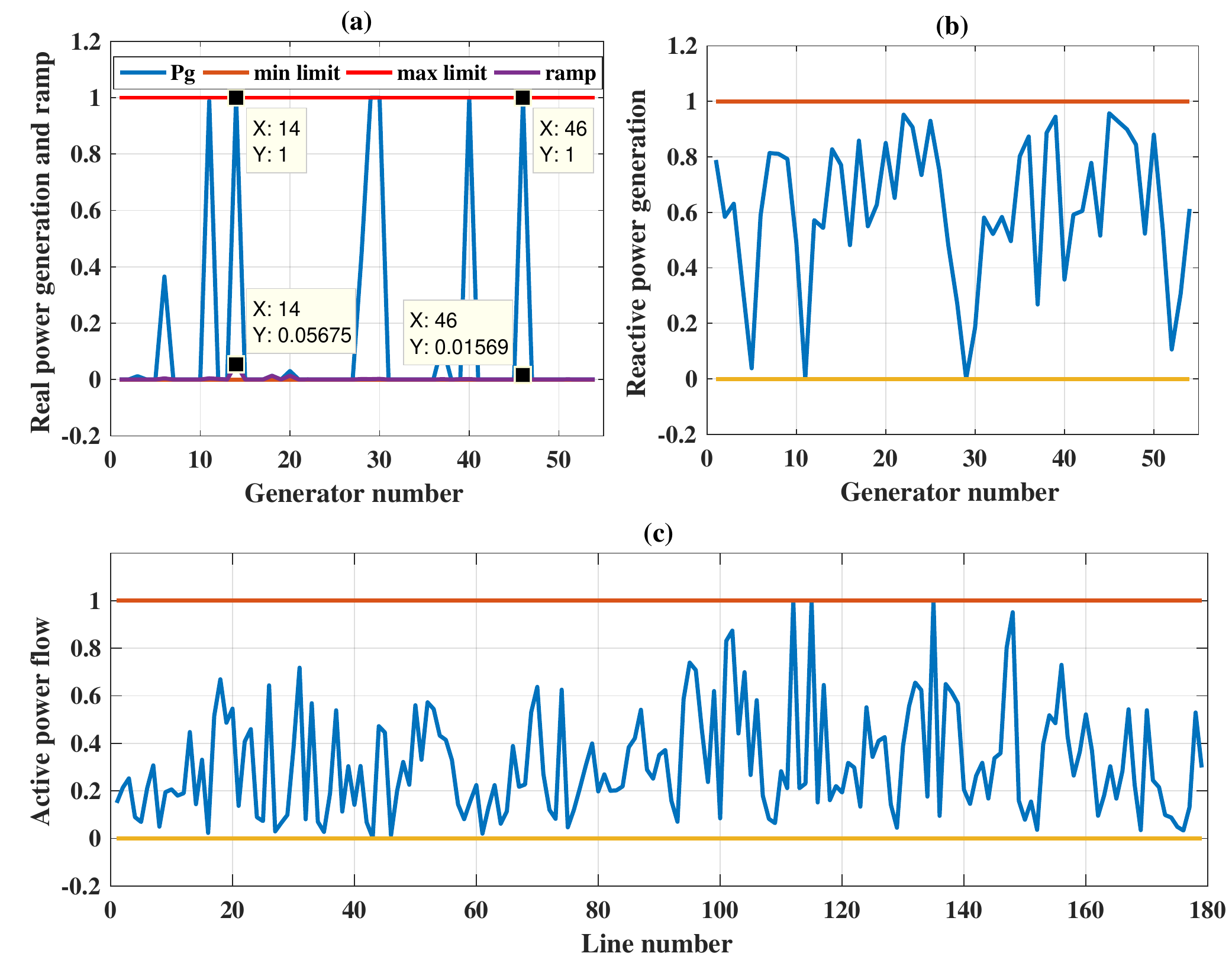}
\caption{Normalized maximum limit, obtained maximum values, and minimum limit of (a) active power generations and generation ramp in \eqref{ramp1}, (b) active power flows, and (c) reactive power generations for $N_s=10000$ scenarios in Algorithm \ref{alg1} for NESTA $118$ bus system with $5\%$ uncertainty in all loads and $15\%$ uncertainty in all RES injections}\label{fig.51}\vspace{-0.5em}
\end{figure}
\begin{figure}[htbp]
\centering\vspace{-0.5em}
\includegraphics[scale=0.7]{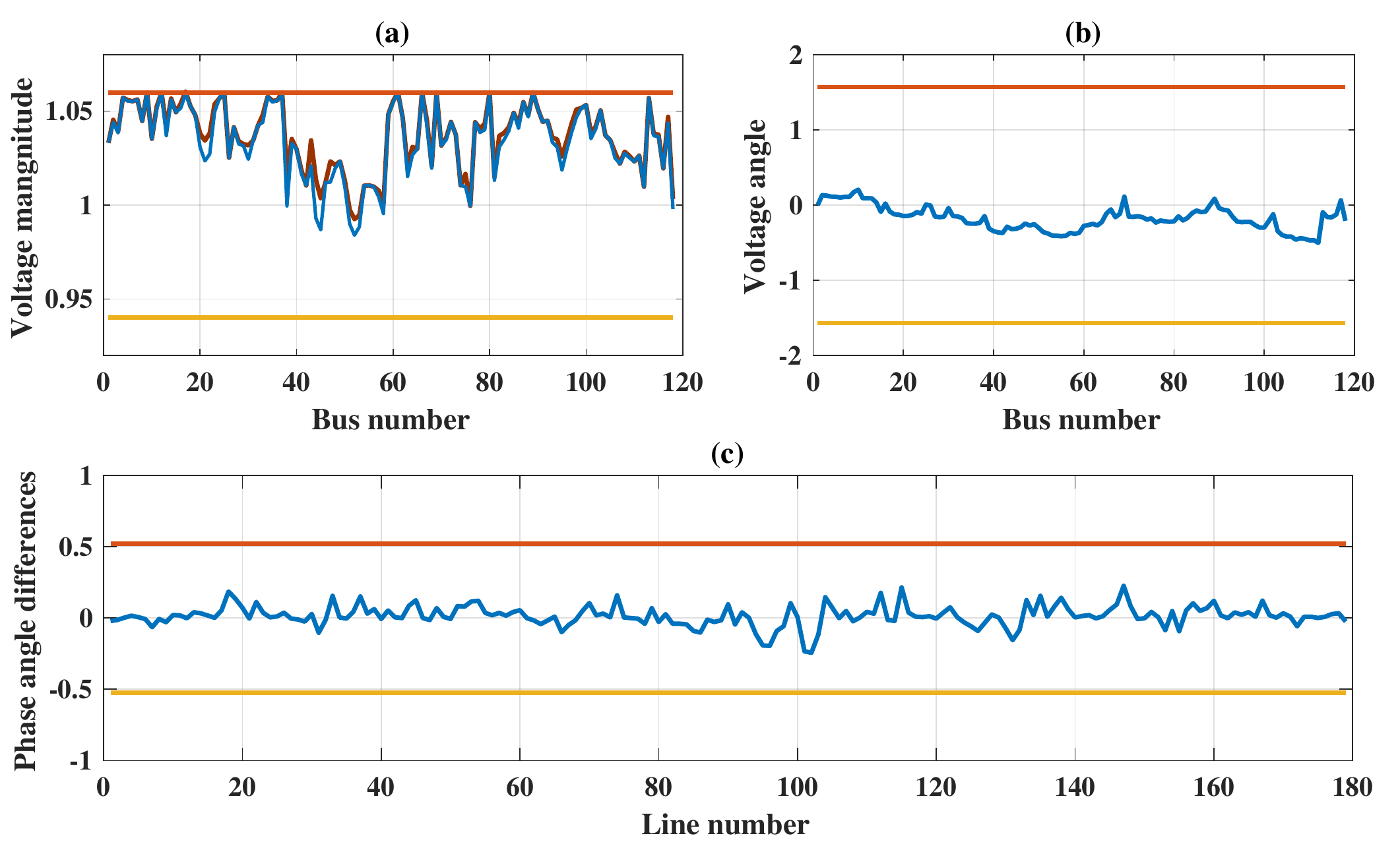}
\caption{Maximum limit, obtained maximum values, and minimum limit of (a) bus voltage magnitudes (pu), (b) bus voltage phase angles (radians), and (c) phase angle differences (radians) across lines for condition in Fig. \ref{fig.51}. The blue line in Fig. \ref{fig.52}(a) shows the obtained minimum bus voltage magnitudes}\label{fig.52}\vspace{-0.5em}
\end{figure}

The robustness of obtained setpoints for unseen realization of uncertainty using out-of-range analysis is also tested. For $\Gamma=|\mathcal{U}|$ and a given $\overline{\bm{\mu}}$ (in Section \ref{sec2b1}), $10000$ random scenarios for out-of-range analysis are generated as $\bm{\mu}\in(-\overline{\bm{\mu}}-0.05,-\overline{\bm{\mu}})\cup(\overline{\bm{\mu}},\overline{\bm{\mu}}+0.05)$. The violation probability in the last column of Table \ref{t3} indicates the \% scenarios in which at least one of the constraints in exact ACOPF is unsatisfied, which is expected as the obtained setpoints are robust only for the given interval of uncertainty, i.e., $\bm{\mu}\in[-\overline{\bm{\mu}},\overline{\bm{\mu}]}) $\cite{robustproof}. Also, with increase in uncertainty level, the violation probability increases due to limited resources in respective AC networks.
\begin{table}[htbp]\vspace{-0.5em}
\centering
\caption{Maximum absolute difference ($\eta$) between robust setpoints from \eqref{m6} and actual ACOPF with uncertain injections set at worst-case for $5\%$ uncertainty in all load and RES injections}
\label{feastable}
\begin{tabular}{cccccc}\hline \hline
Test & NESTA & NESTA & NESTA & NESTA & PEGASE \\
system & 14 & 57 & 118 & 300 & 9241 \\ \hline
 $\eta$ ($\times10^{-5}$) & 0.00492 & 8.50  & 6.56 & 11.3 & 39.5 \\\hline\hline
\end{tabular}\vspace{-0.5em}
\end{table}

The feasibility of the obtained robust setpoints from \eqref{m6} is further verified by comparing the same with the conventional generator setpoints obtained from the actual non-linear and non-convex ACOPF with uncertain injections set at the realized worst-case in \eqref{m6}. The associated maximum absolute differences ($\eta$) for the considered systems with $5\%$ uncertainty in all loads and RES generations is given in Table \ref{feastable}, which indicate that the obtained robust conventional generator setpoints from proposed formulation are also feasible to the actual non-linear and non-convex ACOPF. Since, the setpoints are feasible, the robustness verification via MCS based participation factored AC power flow also yields $0\%$ violation for all in-range scenarios. Similar low values of $\eta$ have also been observed for other cases in Table \ref{t3}, which further support the tightness of \eqref{convcone} and \eqref{23}.

\subsection{Comparison with previous formulations}
\subsubsection{Comparison with deterministic ACOPF}
The solution of deterministic ACOPF in \eqref{model1} is analyzed for injection uncertainties in the NESTA $118$ bus system. Considering nominal values of loads and RES generations, \eqref{model1} is solved using MATPOWER to obtain the conventional generator settings. The optimal objective cost from \eqref{model1} and the associated computation time are given in Table \ref{dettable}.
\begin{table}[htbp]\vspace{-0.5em}
\centering
\caption{Proposed formulation in \eqref{m6} versus deterministic (Det.) ACOPF for NESTA $118$ bus system with $5\%$ load and $15\%$ RES injections uncertainty in Fig. \ref{fig1}}
\label{dettable}
\begin{tabular}{ccc}\hline\hline
Attributes & Det. & Proposed\\\hline
Objective Cost ($\$/hr$) & 2382.38  &  2656.603 \\
Computation time ($s$)  & 0.43 & 1.843 \\
\% violation in active power flow & 100 & 0 \\
\% violation in voltage magnitude & 56.9 & 0 \\ \hline\hline
\end{tabular}\vspace{-0.5em}
\end{table}

For the obtained settings, all constraints of the AC power network are satisfied, if all loads and RES generations are at their respective nominal values. However, for non-zero injection uncertainties, the same generator settings fail to be feasible. This is verified by testing the obtained generator settings (at nominal load and RES generation) against $10,000$ independent and identically distributed scenarios (using Algorithm \ref{alg1}) for $5\%$ uncertainty in all loads and $15\%$ uncertainty in all RES generations. Similar results are shown in Fig. \ref{fig1}, where 6 lines (in red color) exceed their limits by $3\%$ to $48.6\%$ from their maximum capacity for all generated $10,000$ scenarios. Also, bus voltage magnitude constraints are violated for $5690$ scenarios for similar conditions. For the same load and RES injection uncertainties, the optimal objective from \eqref{m6} is $11.51\%$ higher than the optimal cost from \eqref{model1}. However, no constraints' violations are observed for the robust solution obtained from \eqref{m6}.

\subsubsection{Comparison with \cite{bai}}\label{sec4447}
The robust ACOPF formulation in \cite{bai} is also a single-stage formulation. Hence, the solution from proposed formulation in \eqref{m6} is compared with \cite{bai} in terms of objective cost and feasibility for AC power network. Table \ref{baicomptable} shows the comparison between the objective cost of proposed formulation with respect to the objective cost obtained by \cite{bai}, where, for same uncertainty level, \cite{bai} provides more than $5\%$ higher cost than the proposed formulation. \cite{bai} checks for robustness of solution via real-time OPF, which is impractical from the short-term operation perspective. Also, if the validation is to be done by OPF, then the purpose of solving robust ACOPF is lost. Hence, using Algorithm \ref{alg1} with $N_s=10,000$, the robustness of setpoints from \cite{bai} (assuming unity voltage magnitudes) for considered uncertainties in Table \ref{baicomptable} is less than $1.06\%$, whereas from \eqref{m6}, $0\%$ violations for in-range robust analysis is observed in Table \ref{t3}.
\begin{table}[htbp]\vspace{-0.5em}
\centering
\caption{Proposed formulation in \eqref{m6} versus \cite{bai} for NESTA $14$ bus system with $5\%$ in all loads and different levels of RES injections uncertainty}
\label{baicomptable}
\begin{tabular}{cccccc}\hline \hline
$\%$ RES Unc. & $5\%$ & $10\%$ & $15\%$ & $20\%$ \\\hline
$\%$ change in objective cost&  5.5912 &  5.4003  & 5.2083 & 5.0165  \\\hline\hline
\end{tabular}\vspace{-0.5em}
\end{table}

\subsubsection{Comparison with \cite{louca}}\label{sec4446}
The solution from \eqref{m6} is also compared with the solution obtained from the formulation of \cite{louca} in terms of optimal objective, computation time, and load curtailment for the IEEE 14 bus system data given in \cite{louca}. For the same network condition, the robust solution from \eqref{m6} has low cost as compared to the cost reported in Fig. 1(b) on page 10 of \cite{louca}, without any infeasibility until $\sigma=4.5$ ($30\%$ uncertainty in RES) in Fig. \ref{complouca}. For situation just beyond $\sigma=4.5$, exclusive corrective security constrained ACOPF with load curtailment option is solved for the latest feasible trend of the worst-case injection scenario identified from the proposed formulation. The load curtailment penalty is $\$4,000$/MWh, same as in \cite{louca}. The maximum load curtailment from our formulation is $0.0109\%$ of total active power load at non-generator buses, while the same in \cite{louca} is $0.16\%$. Thus, the proposed formulation provides $14.91\%$ load saving. Also, the reported average computation time of formulation in \cite{louca} is $130.4$s, while the same for \eqref{m6} for feasible cases in only $1.367$s (approximately $98.95\%$ less).
\begin{figure}[htbp]\vspace{-0.5em}
\centering
\includegraphics[scale=0.411]{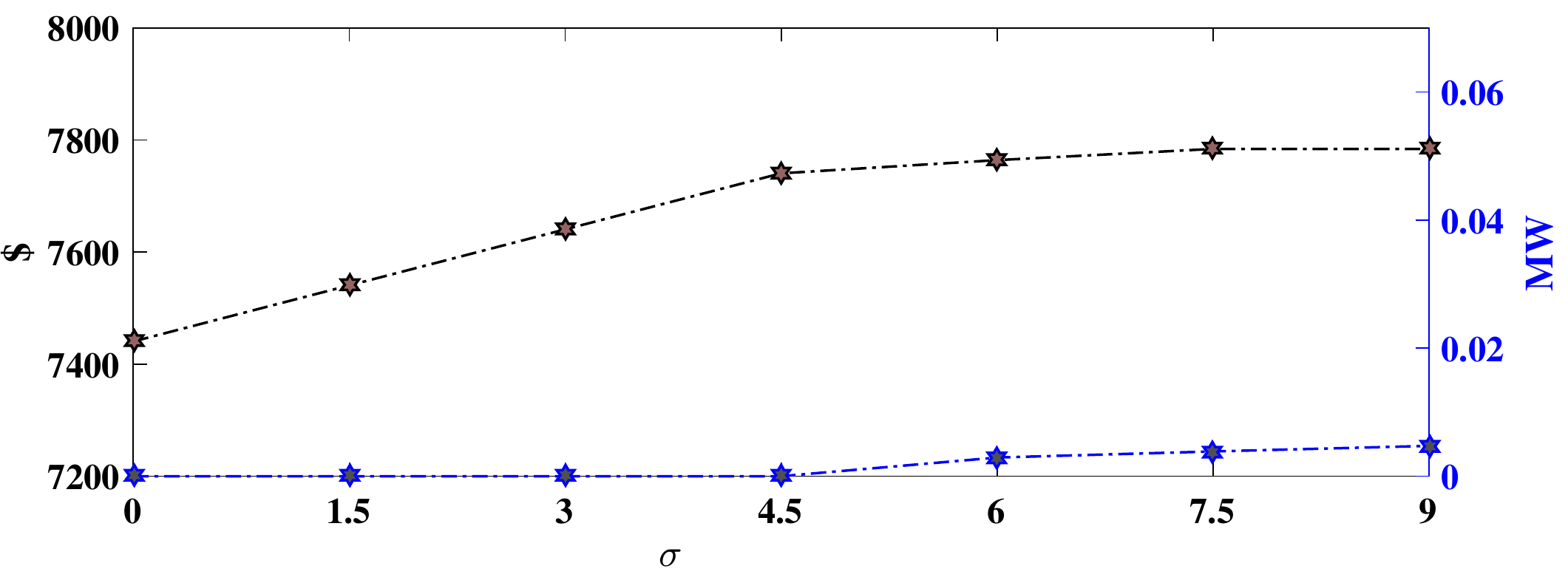}
\caption{Optimal operation cost (black) and load curtailment (blue) from \eqref{m6} for different $\sigma$ tested on IEEE 14 bus system of \cite{louca}}\label{complouca}\vspace{-0.5em}
\end{figure}
\vspace{-0.9em}
\section{Conclusion}\label{sec5}
In this paper, a novel single-stage non-recursive non-integer dual RC of ACOPF is proposed that provides robust setpoints (base-point active power generations and terminal voltage magnitudes) of conventional generators, which are immune to the short-term injection uncertainties in loads and RES generations. By considering constraints for both base-case (no uncertainties) and worst-case uncertainties, and coupling them with ramp constraints on active power generations of conventional generators, the non-linear impact of uncertainties on optimization variables (except conventional generators' setpoints), including change in total active power loss, is inherently ensured without resorting to any affine policies. Using the strong duality property of the novel dual form of convexified ACOPF, novel linear constraints are proposed for the worst-case realization of short-term injection uncertainties. The strong duality of novel dual aids in directly retrieving the desired robust variables in a non-recursive single-stage manner. As compared to previous formulations, the proposed formulation does not rely on cutting plane, mixed-integer or multi-stage approaches. Instead, a novel convexified robust ACOPF model is proposed, which can be solved by existing commercial solvers, such as SeDuMi. The efficiency of proposed formulation is demonstrated through numeric simulations on NESTA systems and the case9241pegase system. For an uncertainty range of $10-15\%$, the obtained worst-case optimal generation cost is $10-25\%$ higher as compared to the nominal generation cost. To reduce this over conservatism, budget of uncertainty constraints are included, which reduce the increase in the generation cost to the range of $5\%$ - $12\%$ of the nominal cost. The obtained robust setpoints are feasible for all realizations of short-term injection uncertainties, which is verified through proofs and MCS based participation factored AC power flows. This verification also justifies the feasibility of obtained setpoints for actual non-linear, non-convex ACOPF with injection uncertainties. In summary, the proposed dual formulation provides time-efficient realization of worst-case uncertainty in the equality constraints of ACOPF. The formulation in \cite{bai} gives higher cost (more than $5\%$) than proposed formulation, with more than $99\%$ probability of violation for in-range robust analysis. Also, compared to \cite{louca}, the proposed formulation has low cost ($14.91\%$ load saving) and computation time ($\approx 98.95\%$ less). It will be of interest to extend the proposed idea to solve robust multi-period security constrained OPF and transmission switching problems.

\section*{Acknowledgement}We are grateful to Science and Engineering Research Board, New Delhi, India for providing financial support to carry out this research work via project no. CRG/2020/001306.

\setcounter{equation}{0}
\renewcommand{\theequation}{\thesection.\arabic{equation}} 
\appendix
\section{Proof of Lemma \ref{theorem1}}\label{appB}
Let $\mathbf{x}^*$ be the Lagrange multiplier of \eqref{m5eqx} obtained from the optimal solution of \eqref{m6}. Since \eqref{model5} and \eqref{m51} are strong primal-dual pairs \cite{yasasvi2021robust}, thus, $\mathbf{x}^*$ along with other Lagrange multipliers are also the optimal solution of \eqref{model5}. Thus, \eqref{model5}, \eqref{m51}, or \eqref{m6} is exact, if for $\mathbf{x}=\mathbf{x}^*$, the following feasibility problem
\begin{align}\label{appeq1}
\max_{\mathbf{y},\mathbf{y}_c}\sum_{i\in\mathcal{G}}Q_{gi}\text{ s.t. } \eqref{m2eq} - \eqref{m2cone2}
\end{align}
has an optimal solution such that \eqref{16} and \eqref{17} are satisfied. A problem similar to \eqref{appeq1} can also be defined in terms of $\mathbf{\hat{y}}$, $\mathbf{\hat{y}}_c$, $\psi$ with $\mathbf{x}=\mathbf{x}^*$, worst-case $\bm{\mu}$ in \eqref{m6}, \eqref{m5eq} - \eqref{m5ramp} as constraints, and objective in terms of $\hat{Q}_{gi}\forall i\in\mathcal{G}$. As per Theorem 3 of \cite{hijazi}, for a network with cycles of size 3, semi-definite programming formulation is equivalent to relaxed version of cyclic constraints combined with SOC constraints. Therefore, under the assumption of presence of flexible generator in a lossless weakly cyclic network, \eqref{appeq1} is same as the one used for proof of Theorem 1 in \cite{probrobust}. Hence, the same reasoning can be used here to show the exactness of \eqref{model5}, \eqref{m51}, or \eqref{m6}.
\bibliographystyle{IEEEtran}
\bibliography{IEEEabrv,mainarxiv.bbl}
\end{document}